\documentclass[runningheads]{llncs}
\usepackage[T1]{fontenc}
\usepackage{graphicx}
\usepackage{caption} 
\usepackage{cite}
\usepackage{amsfonts}
\usepackage{amsmath}
\usepackage{amssymb}
\usepackage{mathtools}
\begin{document}
\title{On the near-tightness of \(\chi \leq 2r\): a general \(\sigma\)-ary construction and a binary case via LFSRs}
\titlerunning{Near-tight \(\chi \leq 2r\): a general \(\sigma\)-ary construction; binary case via LFSRs}
%

\author{Vinicius T.~V.~Date\inst{1} \and Leandro M.~Zatesko\inst{2, 1}}
\authorrunning{V.~T.~V.~Date and L.~M.~Zatesko}

\institute{Postgraduate Program in Informatics, Federal University of Paraná \and Federal University of Technology -- Paraná\\
\email{vtvdate@inf.ufpr.br, zatesko@utfpr.edu.br}}
%
%

\maketitle              
\begin{abstract}
In the field of compressed string indexes, recent work has introduced suffixient sets and their corresponding repetitiveness measure $\chi$. In particular, researchers have explored its relationship to other repetitiveness measures, notably \(r\), the number of runs in the Burrows--Wheeler Transform (BWT) of a string. Navarro et al.~(2025) proved that \(\chi \leq 2r\) \cite{navarro2025smallest}, although empirical results by Cenzato et al.~(2024) \cite{cenzato2024computing} suggest that this bound is loose, with real data bounding \(\chi\) by around \(1.13r\) to \(1.33r\) when the size of the alphabet is \(\sigma = 4\). To better understand this gap, we present two cases for the asymptotic tightness of the \(\chi \leq 2r\) bound: a general construction for arbitrary \(\sigma\) values, and a binary alphabet case, consisting of de Bruijn sequences constructed by linear-feedback shift registers (\mbox{LFSRs}) from primitive polynomials over \(\mathbb{F}_2\). The second is a novel characterization of which de Bruijn sequences achieve the literature run-minimal pattern for the cyclic BWT. Moreover, we show that de Bruijn sequences fail to close the gap for \(\sigma \geq 3\).

\end{abstract}
\section{Introduction}
Repetitiveness measures are linked to compressed indexes. As such, the relationship between different measures is of interest. A recent survey \cite{navarro2021indexing} explores a range of measures and how they are bounded by one another. These relationships between the measures allow us to select the appropriate structure for a given use case.

We are interested in the recently introduced measure \(\chi\) \cite{depuydt2023suffixient}, which is the size of a smallest suffixient set (see technical definitions in Sec.~\ref{sec:prelim}).
We pay close attention to its relationship with \(r\), the number of runs in the Burrows--Wheeler Transform (BWT). A very recent work on suffixient sets established the inequality \(\chi \leq 2r\) (for both the BWT of forward and reverse strings)~\cite{depuydt2023suffixient} \cite{navarro2025smallest}. The bound proposes a factor of $2$, though empirical data for genomes (strings with alphabet size $\sigma =4$) have suggested a lower constant, bounding \(\chi/r\) by around \(1.13\) to \(1.33\) ~\cite{cenzato2024computing}.

Given this discrepancy, we construct a general case for arbitrary alphabet size \(\sigma\) that achieves the ratio $2\sigma/(\sigma+1)$, which approaches the theoretical ratio of $2$ when \(\sigma\) increases. 
On the other hand, for \(\sigma = 2\) (a case of great interest for digital data), we present a run-minimal family using de Bruijn sequences that also approaches the literature bound as the string length increases. This construction matches the run-minimal BWT pattern of \cite{mantaci2017burrows} and is obtained via linear-feedback shift registers (LFSRs) induced by primitive polynomials of the form \(x^k + x + 1\) over \(\mathbb{F}_2\), for \(k\geq 2\).

In Sec.~\ref{sec:prelim} we define the structures and conventions used in the text. In Sec.~\ref{sec:general_family} we present the general \(\sigma\) construction. In Sec.~\ref{sec:binary_family} we present the construction for the binary alphabet. In Sec.~\ref{sec:sigma_geq_two} we explain what happens if we try to generalize the method from the previous section. Finally, in Sec.~\ref{sec:conclusion} we summarize the findings and state open problems and conjectures.

\section{Preliminaries}\label{sec:prelim}

Let \(\Sigma\) be a finite ordered alphabet and let \(\sigma\) denote its size. A string \(w\) over the alphabet \(\Sigma\) is a finite concatenation of characters of \(\Sigma\) and its indexing is \(0\)-based, i.e~\(w \in \Sigma^*\), whose length is denoted \(|w|\); a word of length \(0\) is the \emph{empty} word \(\epsilon\). We also define the \emph{reverse} of \(w\) by \(w^{\mathrm{rev}}[i] := w[|w|-i-1]\) for \(0 \leq i \leq |w|-1\). Concatenation is by juxtaposition: we juxtapose both \(c \in \Sigma\) and \(w \in \Sigma^*\). We define \([i..j]\) as an interval, with its indices being inclusive.

A cycle \(C\) of length \(n\) is a word \(w\) of length \(n\) whose indices are taken modulo \(n\). For a cut position \(c\in [0..n-1]\), the \emph{rotation} of \(C\) by \(c\) is the word \(U\in\Sigma^{*}\) defined by \(U[i]=w[(i+c)\bmod n]\) for \(0 \leq i \leq n-1\).

When needed, we will explicitly terminate a string \(w\) with the symbol \(\$ \notin \Sigma\), \(\$ < a\) for all \(a \in \Sigma\), and \(\$\) will appear once at the end of the string; the resulting representation is explicitly \(w\$\). 

For a string \(w\$\) of length \(n = |w\$|\), the suffix array \(\operatorname{SA}[0..n-1]\) is the permutation of \(\{0,1,\ldots, n-1\}\) such that, for all \(0 \leq i < j < n-1\) 
\[
    w\$[\operatorname{SA}[i]..n-1] < w\$[\operatorname{SA}[j]..n-1] \,,
\] in lexicographic order.

\subsection{Burrows--Wheeler transform}

For a terminated string \(w\$\) of length \(n = |w\$|\), we define the (terminated) Burrows--Wheeler Transform~\cite{burrows1994} of \(w\$\) as follows. The \(\operatorname{BWT}(w\$)\) is the last column \(L\) of the \(n \times n\) matrix whose rows are all cyclic rotations of \(w\$\), sorted lexicographically. Equivalently, let \(\operatorname{SA}[0..n-1]\) be the suffix array of \(w\$\), then
\[
    \operatorname{BWT}(w\$)[i] \ :=\  w\$[(\operatorname{SA}[i]-1)\bmod n]
\]
for \(0\leq i\leq n-1\). We will use both characterizations of the \(\operatorname{BWT}\) throughout the text.

    For strings \(w \in \Sigma^*\) without an explicit endmarker, we use the \emph{circular} Burrows--Wheeler transform, denoted \(\operatorname{cBWT}(w)\). It is defined analogously: consider the \(n\) cyclic rotations of \(w\), sort them lexicographically, and let \(\operatorname{cBWT}(w)\) be the last column of this rotations matrix. 
    
    After either BWT is applied, we can count the number of runs in the resulting string \(L\). A \emph{run} is a maximal equal-letter substring of \(L\). We denote:
    \begin{itemize}
        \item \(r_c\) as the number of runs in  \(\operatorname{cBWT}(w)\);
        \item \(r\) as the number of runs in  \(\operatorname{BWT}(w\$)\); and
        \item \(\bar{r}\) as the number of runs in  \(\operatorname{BWT}(w^{\mathrm{rev}}\$)\).
    \end{itemize}

\subsection{Suffixient set}
    For a given string \(w\), the suffixient set is the set of all one-character right-extensions of right-maximal substrings, as defined bellow\footnote{In \cite{depuydt2023suffixient}, an equivalent definition using suffix trees is also presented.}. The following definitions are borrowed from \cite{navarro2025smallest}, adjusting for 0-based indexing when necessary.
    
    \begin{definition}[Right-maximal substrings and right-extensions]
    Let \(w \in \Sigma^*\). A substring \(x\) of \(w\) is \emph{right-maximal} if there exist distinct \(a,b \in \Sigma\) with both \(xa\) and \(xb\) being substrings of \(w\).
    For any right-maximal substring \(x\) of \(w\), substrings of the form \(xa\) for \(a \in \Sigma\) are called \emph{right-extensions} and their set is given by \(E_r(w) = \{xa \, | \, \exists b \neq a: xa \text{ and } xb \text{ occur in } w \}\).
    \end{definition}
    
    \begin{definition}[Super-maximal extensions]
    Let \(w \in \Sigma^*\). The set of \emph{super-maximal extensions} of \(w\) is defined and denoted by \[S_r(w)  = \{x \in E_r(w) \, | \, \forall y \in E_r, y = zx \implies z = \epsilon \}\,.\] Finally, let \(\operatorname{sre}(w) = |S_r(w)|\).
    \end{definition}
    
    In other words, \(S_r\) is the set of right-extensions that are not the suffix of any other right-extension.

    
    \begin{definition}[Suffixient set]
    Let \(w \in \Sigma^*\). A set \(S\) of indices is suffixient for \(w\) if for every right-extension \(x\) of \(w\), there exists \(j \in S\) such that \(x\) is a suffix of \(w[0..j]\).
    \end{definition}
    
    In our work, we are interested in a smallest suffixient set of a given string \(w\$\), as characterized as follows, with its size being denoted \(\chi\).
    
    \begin{proposition}[Smallest suffixient set]
    Let \(w \in \Sigma^*\). A suffixient set \(S\) is a smallest suffixient set for \(w\) if and only if there is a bijection \(pos : S_r(w) \rightarrow S\) such that every \(x \in S_r(w)\) is a suffix of \(w[0..pos(x)]\).\qed
    \end{proposition}
    
    \begin{definition}[The measure \(\chi\)]
    Let \(w \in \Sigma^*\) and \(S\) a smallest suffixient set for \(w\$\). Then \(\chi(w) = |S|\).
    \end{definition}

     \begin{example}
        Let \(w = aabaa\). We want to compute \(\chi(w)\). To do so, we need to find the size of a smallest suffixient set for \(w\$\). We start by listing every one-character extension for every substring of \(w\$\) (we list only those that occur).
        
        \begin{alignat*}{3}
            &\epsilon:\{a, b,\$\} &&a:\{aa, ab, a\$\} \qquad &&b:\{ba\} \\
            &aa:\{aab, aa\$\} &&ab:\{aba\} &&ba:\{baa\}\\
            &aab:\{aaba\} &&aba:\{abaa\} &&baa:\{baa\$\}\\
            &aaba:\{aabaa\} &&abaa:\{abaa\$\}&&\\ 
            &aabaa:\{aabaa\$\} \quad
        \end{alignat*}

        From the list, we know that the only substrings with \(2\) or more one-character extensions are \(\{\epsilon, a, aa\}\), therefore, they are right-maximal. As such, the set of right-extensions of \(w\$\) is \(E_r(w\$) = \{a, b, \$, aa, ab, a\$, aab, aa\$\}\).

        Now we identify all super-maximal extensions of \(w\). Since we have the set \(E_r\), we remove the right-extensions that are a proper suffix of another, resulting in the set \(S_r = \{aa, aab, aa\$\}\).

        We can finally compute a smallest suffixient set. Let \(S = \{1,2,5\}\). Observe that there is a bijection between \(S\) and \(S_r(w\$)\) such that each \(x \in S_r(w\$)\) is a suffix of \(w\$[0..i]\) for its corresponding index \(i\). Thus, \(S = \{1,2,5\}\) is a smallest suffixient set and \(\chi(w) = |S| = |S_r| = 3\).
   \end{example}

    Note that, when terminating a string \(w\) and attaining the value of \(\chi(w)\), \(\$\) may create a super-maximal extension.

    We are interested in the bound of \(\chi \leq 2r\) \cite{navarro2025smallest}, since it uses the straightforward \(r\) count, as opposed to the earlier bound of \(\chi \leq 2\bar{r}\) \cite{depuydt2023suffixient}.

\subsection{De Bruijn sequences}
     We now define the de Bruijn sequences, which will be used in our result for \(\sigma = 2\). Let \(k \geq 2\). A \emph{de Bruijn sequence} of order \(k\) over $\Sigma$ is a cycle \(C\) of length \(n=\sigma^k\) for which the set of all cyclic length-\(k\) windows of \(C\) (with indices modulo \(n\)) equals \(\Sigma^k\), with every length-\(k\) word occurring exactly once \cite{lothaire1997}.
    
    Since the bound for \(\chi\) expects a linear and terminated word, we associate two types of strings to a given cycle \(C\), using an arbitrary rotation \(U\) of \(C\). We define the linearization of the de Bruijn sequence, which is \(U_\mathrm{lin}:= U\,U[0..k-2]\) (append the first \(k-1\) characters of \(U\)). Then \(|U_\mathrm{lin}| = \sigma^k + (k-1)\) and every length \(k\) word appears exactly once as a contiguous (non-wrapping) substring. The terminated word is the same as defined previously, append \(\$ \notin \Sigma\) to \(U\), resulting in \(U\$\).
    
    We use ``de Bruijn sequence of order \(k\)'' to refer to the rotation of the \emph{cycle} \(C\) that starts with \(a_{\text{min}}^k\), where \(a_{\text{min}}\) is the smallest symbol in \(\Sigma\), which is well-defined since every word of length \(k\) appears exactly once. 
    
    \begin{example}
        One possible de Bruijn sequence of order \(3\) over \(\sigma = 2\) with \(\Sigma = \{0, 1\}\) is \(U = 00010111\). The set of all cyclic length \(3\) windows equals \[
            \{000, 001, 010, 101, 011, 111, 110, 100\} \,.
        \]
        
        The linearized form of \(U\) is \(U_{lin} = 0001011100\), and its terminated version is \(U\$ = 00010111\$\).
    \end{example}
    
\subsection{From LFSRs to de Bruijn sequences}
    The linear-feedback shift register (LFSR) is our choice of de Bruijn-sequence construction. An LFSR uses a \emph{characteristic polynomial} over a finite field to induce a recurrence that will be used to create an \(m\)-sequence, as defined below \cite{golomb1981shift,lidl1997finite,goresky2012algebraic}. With the cycle-joining method, we arrive at a valid de Bruijn sequence \cite{etzion1984algorithms}.
    
    Since we are interested in \textit{binary} de Bruijn sequences, the characteristic polynomials will be taken over \(\mathbb{F}_2=(\{0,1\},+,\cdot)\).
    
    Let \(C(x)\) be a characteristic polynomial of the form
    \[
        C(x)=x^{k}+c_{k-1}x^{k-1}+\cdots+c_{1}x+c_{0}\in\mathbb{F}_2[x]\,.
    \]
    We say that \(C(x)\) induces the recurrence \(s_{t+k} =\sum_{i=0}^{k-1} c_i\,s_{t+i} \pmod{2}\); equivalently, the recurrence is the XORing of the tapped bits with \(c_i = 1\): \(s_{t+k} =\bigoplus_{i=0}^{k-1} c_i\,s_{t+i}\). At time \(t\), the \emph{state} of the recurrence is the \(k\)-tuple 
    \[
        x_t=(s_t,\ s_{t+1},\ \dots,\ s_{t+k-2},\ s_{t+k-1})\in\{0,1\}^k \,;
    \]
    the next state of the recurrence is 
    \[
        x_{t+1}=(s_{t+1},\ s_{t+2},\ \dots,\ s_{t+k-1}, s_{t+k})\, .
    \]
    Finally, define the successor map \(F(x_0,\ldots,x_{k-1}):=(x_1,\ldots,x_{k-1},\,\bigoplus_{i=0}^{k-1}c_i x_i)\); thus \(x_{t+1}=F(x_t)\).
    This is known as the \emph{Fibonacci form}: we left-shift the sequence and append \(s_{t+k}\) \cite{goresky2012algebraic}.
    
    We are interested in characteristic polynomials \(C(x)\) of degree \(k\) that are \emph{primitive} over \(\mathbb{F}_2\), that is, \(C(x)\) is irreducible and the residue class \(\alpha := x \bmod C(x)\) has multiplicative order \(2^k-1\) in \(\mathbb{F}_2[x]/(C(x))\) \cite{lidl1997finite}. Primitive polynomials coupled with a non-zero seed produce maximum period sequences, called \emph{\(m\)-sequences}, of period \(2^k-1\), with the notable exception being the all-zero state \cite{golomb1981shift}. In fact, all non-zero seeds yield the same \(m\)-sequence up to rotation.
    
    In other words, \(m\)-sequences produce a period that visits every non-zero \(k\)-tuple state exactly once before repeating. Should we use the all-zero seed, the recurrence would lock into the all-zero state. Therefore, we can say that a recurrence induced by a primitive polynomial produces two cycles: one for the \(m\)-sequence (main cycle) and another for the all-zero self-loop (\(0\)th cycle).
    
    To have a true de Bruijn sequence, we employ the \textit{cycle-joining} operation \cite{etzion1984algorithms}. The cycle-joining gives us a way to join the main cycle and the \(0\)th cycle, forming a single cycle of length \(2^k\), a binary de Bruijn sequence.
    
    To cycle join, we first find states from two disjoint cycles that share the same suffix of length \(k-1\), \(b = (b_0, b_1,\dots, b_{k-2})\). Let \(U,V\) be disjoint cycles from our binary recurrence, let \(u = (1, b) \) be a state of \(U\) and let \(v = (0, b)\) be a state of \(V\). Formally, we call \(p = (u,v)\) a \textit{conjugate pair} if they share the suffix tuple \(b\).
    
    With a conjugate pair \(p = (u,v)\) found, we now \textit{join} the cycles. To join them, both states from \(p\) need to \textit{interchange their successors}. Let \(F\) be the successor map. For a conjugate pair \(p = (u,v)\), \(F'(u) = F(v), F'(v) = F(u)\) and \(F'(x) = F(x)\) for all other states \(x\). This modification is local: only the two joined states differ from the original linear recurrence; all other transitions are unchanged.

    \begin{example}
        The recurrence induced by the primitive polynomial \(x^3 + x + 1\) has two cycles. Using the seed \((s_0, s_1, s_2)= (0, 0, 1)\), the main cycle (up to rotation) is \(0010111\) and the \(0\)th cycle is \(000\).
    
        The conjugate pair in our case is \(p = (000, 100)\), since they share the same suffix \(00\). Let us apply the successor swap, \(F'(000) = F(100) = 001\) and \(F'(100) = F(000) = 000\). The resulting de Bruijn sequence of order \(3\) is \(00010111\), which, by definition, has all binary \(3\)-tuples.
    \end{example}

    Finally, we interpret a given period as a word. Let \(n = 2^k\) denote the period after cycle-joining, i.e.~a de Bruijn sequence. We form the string \(w\in\{0,1\}^n\) by setting \(w[i]=s_i\) for \(0 \leq i<n\). We view \(w\) cyclically, i.e.~the \(k\)-length window at a position \(t\) is \(w_t = w[t..t+k-1 \bmod{n}]\), therefore \(w_t = x_t\).

\section{A family of strings with \(\chi \approx 2r\)}\label{sec:general_family}
In this section, we will present a family that gets arbitrarily close to the upper bound of \(\chi \leq 2r\) by \cite{navarro2025smallest} as the size of the alphabet increases.

Let \(\mathcal{K}\) be the \emph{clustered} family of strings, as defined below. We show that such family both maximizes \(\chi\) and minimizes \(r\) when \(\$\)-terminated. A string \(K \in \mathcal{K}\) is defined as 
\[
K := s_{n-1}^{k_{n-1}}\, s_{n-2}^{k_{n-2}}\, \dots\, s_{1}^{k_1}\, s_{0}^{k_0} \,,
\]
with \(s_k \in \Sigma\), \(s_0 < s_0  < \dots  < s_{n-1}\) and \(k_i > 1, \forall i\). That is, elements of \(\mathcal{K}\) are words of equal-letter runs sorted by the inverse lexicographical ordering. 

Now, we show the value of \(r\) attained by words \(K \in \mathcal{K}\).
\begin{lemma}\label{cluster_runs}
    Let \(K \in \mathcal{K}\). When applying the BWT to \(K\$\), it results in \(r = \sigma + 1\) runs.
\end{lemma}
\begin{proof}
    We will examine the lexicographic ordering of the suffixes of the BWT rotation matrix. This ordering groups suffixes into \emph{buckets}. For the string \(K\$\), we have the single \(\$\)-bucket, then the \(s_0\)-bucket, up until the \(s_{n-1}\)-bucket. Except for the \(\$\)-bucket, every bucket has at least two suffixes.

    For \(j \in \{0,\dots,n-1\}\), each suffix of a \(s_j\)-bucket has a predictable form, it starts with the prefix block \(s_j^{i}\), for some \(1\leq i\le k_j\), and is followed by a strictly smaller symbol \(s_{j-1}\) (\(\$\) when \(j = 0\)). Therefore, we can see that the suffixes of a given \(s_j\)-bucket are ordered by increasing \(i\), i.e.~the shorter prefixes of \(s_j\) come first.

    After describing the ordering of a given bucket, we examine the values for the last column \(L\) of the rotation matrix. First, the \(\$\)-bucket is special, it only has the \(\$\) suffix, as such, \(s_0\) is its preceding character. Consequently, we examine each suffix of every other bucket \(s_j\). For every suffix of a \(s_j\)-bucket, the preceding character is \(s_j\) itself, with the exception of the suffix \(s_j^{k_j}\), whose preceding character is \(s_{j+1}\), or \(\$\) if \(j = n-1\).

    Visually, using the ordering described above, \(L\) is given by
    \[
    L =
    \underbrace{s_0}_{\$\text{-bucket}} \cdot
    \underbrace{s_0^{k_0-1}\ s_1}_{s_0\text{-bucket}} \cdot
    \underbrace{s_1^{k_1-1}\ s_2}_{s_1\text{-bucket}} \cdots
    \underbrace{s_{n-1}^{k_{n-1}-1}\ \$}_{s_{n-1}\text{-bucket}}
    = s_0^{k_0}\, s_1^{k_1}\, \dots \, s_{n-1}^{k_{n-1}}\, \$.
    \]
    
     We conclude that \(L\) has \(\sigma+1\) runs, one for each distinct character of \(K\$\).\qed
\end{proof}

Similarly, we show the size of the smallest suffixient set \(\chi\) for \(K\$\).

\begin{lemma}\label{cluster_suffixient}
    Let \(K \in \mathcal{K}\). A smallest suffixient set for \(K\$\) has size \(\chi = 2\sigma\).
\end{lemma}
\begin{proof}
    To determine the value of \(\chi\), we count the number of super-maximal extensions of \(K\$\). For \(j \in \{0, \dots, n-1\}\), we will look at each \(s_j\) cluster and characterize their super-maximal extensions.

    For a given cluster \(s_j\), we list all right-maximal substrings and their right-extensions. Every substring of the form \(s_j^i\), for \(0 \leq i \leq k_j-1\), has two right-extensions, namely \(s_j^i \, s_j\) and \(s_j^i \, s_{j-1}\). Therefore, every substring of this form is right-maximal, and they are the only right-maximal substrings.
    
    The exceptions are the substrings of the form \(s_j^{k_j}\), such substrings cannot possibly have the character \(s_j\) as a right-extension, therefore, they are not right-maximal.

    Now, we filter the right-extensions and identify those that are super-maximal. We can see that we only have two super-maximal extensions per \(s_j\) cluster. Every other right-extension is either a suffix of \(s_j^{k_j} = s_j^{k_j-1} \, s_j\) or a suffix of \(s_j^{k_j-1} \, s_{j-1}\), with \(s_{j-1} = \$\) if \(j = 0\).

    We conclude that a \(s_j\)-cluster of \(K\$\) has only \textit{two} super-maximal extensions.

    Hence, the smallest suffixient set of \(K\$\) has \(2\) positions per cluster, resulting in a total of \(2\sigma\) positions.\qed
\end{proof}

\begin{example}
    We have the string \(K = 332222111,\, K \in \mathcal{K}\). To count both \(r\) and \(\chi\) we need to terminate it, obtaining \(K\$ = 332222111\$\). We start by examining the rotation matrix of \(K\$\)
    \begin{align*}
        &\$332222111\\
        &1\$33222211\\
        &11\$3322221\\
        &111\$332222\\
        &2111\$33222\\
        &22111\$3322\\
        &222111\$332\\
        &2222111\$33\\
        &32222111\$3\\
        &332222111\$ \,.
    \end{align*}    
    The last column \(L\) is \(111222233\$\), which contains \(4\) equal-letter runs.
    
    Now we list the super-maximal extensions of \(K\$\)
    \[
        S_r = \{33,32,2222,2221,111,11\$\}\, ,
    \] therefore, \(\chi = 6\).
\end{example}

\begin{theorem}
    As the alphabet size \(\sigma\) increases, the ratio \(\chi/r\) for the clustered family \(\mathcal{K}\) of strings approaches two.
\end{theorem}
\begin{proof}
    
    From Lemmas~\ref{cluster_runs} and \ref{cluster_suffixient},
    \begin{equation*}
        \lim_{\sigma \rightarrow \infty }{\frac{\chi}{r}} = 
        \lim_{\sigma \rightarrow \infty }\frac{2\sigma}{\sigma + 1} = 
        2\,.
        \eqno \qed
    \end{equation*}
    
\end{proof}

\section{Tightening the gap for binary \(\sigma\)}\label{sec:binary_family}

In the previous section, we described a family of strings that reaches the proposed \(\chi \leq 2r\) bound for large values of \(\sigma\). However, for small values, we cannot reach the factor of two runs per suffixient set position. For example, if we take \(\sigma = 2\), strings of the family \(\mathcal{K}\) have a bound of \(4/3\), way below the desired \(2r\).

A natural follow-up question is to ask whether a construction for smaller \(\sigma\) values exists. In this section, we partially address this question by presenting a family of strings with \(\sigma = 2\) that also approaches the upper bound as the string length increases.

We will use the fact that linearized binary de Bruijn sequences \textit{maximize} the number of indices in the suffixient set \cite{navarro2025smallest}. For such words, the suffixient set size remains the same up to rotation, reversal, and complement, therefore, we have more flexibility when optimizing for the number of runs \(r\).

To describe which de Bruijn sequences can \textit{minimize} \(r\), we use \mbox{LFSRs} induced by primitive polynomials of a specific form. Once we have such sequences, their cBWT matches the result from \cite{mantaci2017burrows}, i.e.~the cBWT achieves the run-minimal pattern for binary de Bruijn sequences.

The final step is to both linearize and terminate the de Bruijn sequences obtained. We show how to count the number of runs in the BWT, reaching the bound of \(\chi / r = (2^k+1)/(2^{k-1}+4)\). 

\subsection{From Polynomials to LFSRs}\label{subsec:trinomials}

We begin by specifying which polynomials are used to induce the LFSR; we will work with polynomials \(T_k \in \mathcal{T}\), for \(k \geq 2\), of the form 
\[
   T_k(x) = x^k + x + 1
\]
that are primitive over \(\mathbb{F}_2\). The recurrence induced by such polynomials is 
\[
    s_{t+k} = s_{t+1} \oplus s_t\,.
\]

The degrees \(k\) for which \(T_k\) is primitive are exactly those listed in OEIS A073639 \cite{OEIS:A073639}. To the best of our knowledge, there is no proof that this sequence is infinite.


\subsection{Cycle-joining the 0th cycle and its effects on the sequence} \label{subsec:cycle_joining}

For a non-zero seed, a linear recurrence induced by a primitive polynomial has a period of \(2^k-1\); we are one state short of a true de Bruijn sequence. Hence, we need to amend our construction. In this section, we employ the \textit{cycle-joining method} described in Sec.~\ref{sec:prelim}.

We start by identifying two disjoint cycles. Since the LFSR induced by a primitive polynomial results in exactly two cycles, we use those two; we choose the main and the \(0\)th cycles. 

The next step is to identify the conjugate pair \(p\). The pair \(p\) is composed of one state from each cycle; such states share the same suffix \(b\), with \(|b| = k-1\). Since the \(0\)th cycle is a self-loop of the state \(0^{k}\), we conclude that the target suffix must be \(b = 0^{k-1}\). From \(b\), we identify that the other state is \(10^{k-1}\). Thus, \(p = (0^{k},10^{k-1})\).

Finally, with \(p\) found, we perform the cycle-joining method. We start by stating the original successors for the pair \(p\): \(F(0^k) = 0^k\) and \(F(10^{k-1}) = 0^{k-1}1\). Following the method, we swap them, resulting in the transitions \(F'(0^k) = 0^{k-1}1\), \(F'(10^{k-1}) = 0^k\), \(F'(x) = F(x)\), for every other state \(x\). The overall result is a sequence of period \(2^k\), a de Bruijn sequence \cite{etzion1984algorithms}.

Before we proceed, we need to patch up the recurrence accordingly. For the identified conjugate pair \(p = (u,v) = (0^k, 10^{k-1})\), we complement their resulting states in the recurrence, which is equivalent to XORing them by \(1\):

\[
    s'_{t+k} = s_{t+1} \oplus s_t \oplus (1_{x_t=u} \oplus 1_{x_t=v}) \,,
\]
where \(1_{\{\text{statement}\}}\) is the indicator function, which \(1\) if the statement is true, and \(0\) otherwise.

From now on, we denote \((1_{x_t=u} \oplus 1_{x_t=v})\) simply by \(\omega\).

\begin{example}
    For \(k = 3\) and seed \((s_0, s_1, s_2) = (0,0,1)\), the altered recurrence \(s'\) produces the \(2^3\) state cycle \(00010111\), which is a valid de Bruijn sequence of order \(3\).
\end{example}

When applying the cycle-joining method, we make sure to remember which states compose the conjugate pair \(p\). In particular, our construction always joins the same states, \(u = 0^k\) and \(v = 10^{k-1}\).

\subsection{Transforming the sequence}\label{subsec:transformation}
Finalizing the de Bruijn sequence construction, we aim to achieve the run-minimal pattern in the cBWT of \(1(0011)^{(2^{k-2}-1)}010\) \cite{mantaci2017burrows}. To achieve this, we apply operations that do not change the de Bruijn property, but do alter the cBWT: we reverse and complement the generated de Bruijn sequences. The resulting set of sequences is \emph{the run-minimal family} of cyclic words \(\mathcal{M}\). Elements \(M_k \in \mathcal{M}\) are indexed by the same \(k\) as elements \(T_k \in \mathcal{T}\).

We start by exposing how the operations alter the underlying recurrence for any LFSR. Let be \(C\) be a primitive polynomial and let \(s_{t+k} =\bigoplus_{i=0}^{k-1} c_i\,s_{t+i}\) be its recurrence.

To analyze the reversal operation, we look at the \emph{reciprocal polynomial} of \(C\), denoted \(C^*\), since it preserves primitivity \cite{goresky2012algebraic}. Let \(C^*(x) = x^kC(1/x) = 1 + c_{k-1}x + \ldots + c_1 x^{k-1} + c_0 x^k\). The effect is the mirroring of the tapped bits
\[
    r_{t+k} = \bigoplus_{i = 0}^{k-1} c_{k-1-i}r_{t+i}\,.
\]

Next, we examine how the complement alters the underlying recurrence. Complementing each element \(s_t\), since we are working on \(\mathbb{F}_2\), is the same as XORing each element by \(1\). Let \(p_t := s_t \oplus 1\), then

\[
    p_{t+k} = \bigoplus_{i = 0}^{k-1} c_ip_{t+i} \ \oplus \ \bigoplus_{i=0}^{k-1}c_i \ \oplus \ 1 \,.
\]

The overall effect on the recurrence is decided by the parity of tapped bits. If the parity is even, we XOR the right side by \(1\), otherwise the recurrence does not change \cite{golomb1981shift}.

Finally, we show how our particular recurrence \(s_{t+k} = s_{t+1} \oplus s_t\) changes with these operations.

First, we apply the reverse operation. The reciprocal polynomial of \(T_k \in \mathcal{T}\) is 
\[
    T^*_k(x) = x^k T_k(1/x) = x^k(x^{-k}+x^{-1} + 1) = 1 + x^{k-1} + x^k \,,
\]
which results in the recurrence \(s_{t+k} = s_{t+k-1} \oplus s_t \oplus \omega\).

The next step is to complement it. Since the reversal does not alter the number of taps, the recurrence still has an even number of them. Therefore, we can XOR the reversed recurrence by \(1\), resulting in

\[
    s_{t + k} = s_{t+k-1} \oplus s_t \oplus 1 \oplus \omega \,.
\]

Note that the conjugated pair \(p\) is altered as well. State \(0^k\) changes into \(1^k\) and state \(10^{k-1}\) changes into \(1^{k-1}0\). Therefore \(\omega = (1_{x_t = 1^k} \oplus 1_{x_t=1^{k-1}0})\).

\subsection{Achieving the run-minimal pattern}\label{subsec:run_min_construct}
Finally, we can analyze the cBWT of \(M_k \in \mathcal{M}\). In this section, we find a closed form for the last column \(L\) of the rotation matrix for every \(M_k \in \mathcal{M}\). 

We start by describing a pattern in the rotation matrix of any de Bruijn sequence. Since the number of suffixes of a de Bruijn sequence of order \(k\) is equal to the size of the sequence itself, we have \(2^k\) rows in the rotation matrix. Therefore, we can sort the rotation matrix by only the first \(k\) symbols of each row, i.e.~the matrix is solely ordered by the \(k\)-length states.

Knowing that, we manipulate the recurrence in such a way that we can infer the column \(L\) by the first \(k\) symbols of any line:

\begin{align*}
    s_{t+k} &= s_{t+k-1} \oplus s_{t} \oplus 1 \oplus \omega\\
    s_{t} &= s_{t+k} \oplus s_{t+k-1} \oplus 1 \oplus \omega \,.
\end{align*}

While the recurrence originally determined the current state by looking backwards, its altered form determines the current state by looking \emph{forwards}, which, for the last column \(L\), results in wrapping around, looking at the starting symbols of every row. Furthermore, the values \(s_{t+k-1}\) and \(s_{t+k}\) are the rightmost characters of any \(k\)-length state, therefore they are also the least significant digits of such states. As such, they repeat every \emph{four} rows, as does the resulting column \(L\):

\begin{align*}
    s_{0}\ s_{1} \cdots s_{t+k-2}\ 00& \cdots 1\\
    s_{0}\ s_{1} \cdots s_{t+k-2}\ 01& \cdots 0\\
    s_{0}\ s_{1} \cdots s_{t+k-2}\ 10& \cdots 0\\
    s_{0}\ s_{1} \cdots s_{t+k-2}\ 11& \cdots 1\\
    \vdots&
\end{align*}

From this regularity, we can deduce that the column \(L\) is \((1001)^{2^{k-2}}\). But this is not the final result, since we still need to amend the values from the cycle joined states, i.e.~we need to account for \(\omega\) in the recurrence.

As a reminder, \(1^{k-1}0\) and \(1^k\) are our cycle joined states. These states correspond to the last two rows in our rotation matrix, therefore, when accounting for \(\omega\), we get the following for \(L\):

\begin{alignat*}{3}
    &&\vdots\\
    s_{0}\ s_{1}& \cdots &00&& \cdots 1\\
    s_{0}\ s_{1}& \cdots &01&& \cdots 0\\
    11& \cdots &10&& \cdots \textbf{1}\\
    11& \cdots &11&& \cdots \textbf{0}
\end{alignat*}

We arrive at \(1(0011)^{2^{k-2}-1}010\), the run-minimal pattern for cBWTs \cite{mantaci2017burrows}. Its run count \(r_c\) is \(2^{k-1} + 2\), where the \(+2\) corresponds to the cycle joined states.

\begin{theorem}
    For every \(k\ge 2\) such that \(T_k \in \mathcal{T}\) is primitive over \(\mathbb{F}_2\), let \(\mathcal{M}\) be the family of de Bruijn sequences produced by cycle joining LFSRs induced by \(T_k\) after the reversal and complement operations. Strings \(M_k \in \mathcal{M}\) have a cBWT of \(1(0011)^{2^{k-2}-1}010\), the run-minimal pattern.
\end{theorem}
\begin{proof}
    By Subsections~\ref{subsec:trinomials}, \ref{subsec:cycle_joining}, \ref{subsec:transformation}, and \ref{subsec:run_min_construct}.\qed
\end{proof}


\subsection{The case for linearization and \$-terminated strings}
Aiming to bring \(M_k \in\mathcal{M}\) into the suffixient set domain, we both linearize and terminate them. After both operations, we count the number of runs \(r\) and the size of the smallest suffixient set \(\chi\), reaching our central result of \(\chi/r = (2^k + 1)(2^{k-1} + 4)\).

We start by choosing a specific rotation \(U\) of \(M_k \in \mathcal{M}\). We select the rotation starting with \(0^k\). The linearization and termination process is: we start with the rotation \(U = 0^k\dots1^k\), then we linearize and terminate it, resulting in the string \(U_\mathrm{lin}\$ = 0^k \dots 1^k0^{k-1}\$\).

The following lemma helps us understand how the operations affect the BWT.

\begin{lemma}\label{lemma:bwt_runs}
    Let \(U\) a be a Bruijn sequence of order \(k\), with \(U\) being a rotation of \(M_k \in \mathcal{M}\). When we linearize and terminate \(U\), the BWT of \(U_\mathrm{lin}\$\) is \(0^{k-1}1\$(0011)^{2^{k-2}-1}010\). And its number of runs is \(r = 2^{k-1} + 4\).
\end{lemma}
\begin{proof}
    Let \(L\) be the last column of the rotation matrix of \(U_\mathrm{lin}\$\). We will divide the proof into two parts, the \(k\)-length suffixes containing \(\$\) and those without it.
    
    As mentioned previously, the rotation matrix of \(U_\mathrm{lin}\$\) corresponds to the ordering of suffixes of \(U_\mathrm{lin}\$\). Therefore, the first suffix is \(\$\), since it is defined to be smaller than every character in \(\Sigma\).
    
    We start by looking at the remaining suffixes containing \(\$\) of length at most \(k\). Such suffixes are of the form \(0^j\$\), with \(1 \leq j \leq k-1\). They form the start of the rotation matrix, and their ordering is: \(0\$, \dots 0^{k-1}\$\). 
    
    With that set, we inspect their preceding characters to form the prefix of \(L\). For every suffix containing \(\$\), with the exception of \(0^{k-1}\$\), the preceding character is \(0\). As for the suffix \(0^{k-1}\$\), we know, by construction of linearization, that the preceding character is \(1\).

    Putting all together, the prefix \(x\) of \(L\) is
    \[
        x = 0^{k-1}1 \,.
    \]

    Now, we move on to the unchanged suffixes. They preserve the ordering of the cBWT of \(M\), since \(\$\) is out of reach for tie breaking. 
    Similarly, the remaining characters of \(L\) are mostly the same, with the only change being the preceding character of the \(0^k\) suffix, which is mapped to \(\$\). We can see it as the \(1\) from the cBWT and the newly inserted \(\$\) changing places in the BWT.

    As such, the suffix \(y\) of \(L\) is 
    \[
        y = \$(0011)^{2^{k-2}-1}010 \,.
    \]

    Concatenating \(x\) and \(y\), we complete \(L\):
    \[
        L = 0^{k-1}1\$(0011)^{2^{k-2}-1}010 \,.
    \]
    
     Finally, having \(L\), we count the number of runs \(r\) in the BWT, which equals to \( 2^{k-1} + 2 + 2 = 2^{k-1} + 4\).\qed
\end{proof}

What is left is to show how the construction affects \(\chi\), the size of a smallest suffixient set. In the following lemma, we prove that \(\chi = 2^k + 1\), using the result from \cite{navarro2025smallest} as a basis. In their paper, the authors show that the \(\operatorname{sre}\) (the number of super-maximal extensions) for linearized and unterminated binary de Bruijn sequences \(w_\mathrm{lin}\) is \(\operatorname{sre}(w_\mathrm{lin}) = 2^k\). We expand their result for \(\$\)-terminated de Bruijn sequences.

\begin{lemma}\label{lemma:suffixient_set_positions}
    The size \(\chi\) of a smallest suffixient set for a linearized and terminated binary de Bruijn sequence \(w_\mathrm{lin}\$\) of order \(k\) is \(2^{k}+1\).
\end{lemma}
\begin{proof}
    Let \(w_\mathrm{lin}\) be a linearized but unterminated binary de Bruijn sequence of order \(k\). From \cite{navarro2025smallest} we know that \(\operatorname{sre}(w_\mathrm{lin}) = 2^k\). 
    
    When we terminate \(w_\mathrm{lin}\), the suffix \(s\) of length \(|s| = k-1\) from \(w_\mathrm{lin}\) forms an extra right-extension, namely \(s\$\). The suffix \(s\) already had two right-extensions, \(s0\) and \(s1\). Therefore \(s\) was already a right-maximal substring. 
    
    All three right-extensions of \(s\) are super-maximal, not occurring anywhere else in the string, with \(s0\) and \(s1\) occurring exactly once because of the de Bruijn property, and \(s\$\) occurring exactly once because \(\$ \notin \Sigma\).
    
    Since both \(s0\) and \(s1\) are already covered, we only need to cover \(s\$\), resulting in a single extra position in the smallest suffixient set. Therefore, \(\chi = 2^k+1\).\qed
    
\end{proof}


\begin{theorem}
    Let \(k \geq 2 \) be and integer such that \(T_k \in \mathcal{T}\) is a primitive polynomial over \(\mathbb{F}_2\). Let \(U\) be a rotation of \(M_k\) that starts with \(0^k\). The resulting ratio of \(\chi / r\) concerning \(U_\mathrm{lin}\$\) approaches \(2\) as \(k\) increases.
\end{theorem}
\begin{proof}
    From Lemmas~\ref{lemma:suffixient_set_positions} and \ref{lemma:bwt_runs} we know that \(\chi = 2^k + 1\) and \(r = 2^{k-1} + 4\). Since the size of \(U_\mathrm{lin}\$\) is directly proportional to \(k\) by construction, if we take the limit of \(\chi / r\) as \(k\) approaches infinity we obtain the bound
    \[
        \lim_{k \rightarrow \infty}{\frac{\chi}{r} = \lim_{k \rightarrow \infty}\frac{2^k + 1}{2^{k-1} + 4} = 2}\,.\eqno \qed
    \]
\end{proof}

Although we take the limit as \(k\) goes to infinity, recall that we do not know if the sequence of origin (the values for \(k\) such that \(T_k\) is primitive) is infinite.

\section{De Bruijn sequences for \(\sigma > 2\)}\label{sec:sigma_geq_two}
In this section, we explain how de Bruijn sequences with \(\sigma > 2\) cannot approach the literature ratio of \(\chi / r \leq 2\); as we will see, when the value of \(\sigma\) increases, the ratio converges to \(1\).

We start by generalizing Lemma 5 of \cite{navarro2025smallest}, which gives us the number of super-maximal extensions for a non-terminated binary string. We observe that the same proof works for arbitrary values of \(\sigma\).

\begin{lemma}
    Let \(w_\mathrm{lin}\) be a linearized \(\sigma\)-ary de Bruijn sequence rotation of order \(k > 0\). The length of \(w_\mathrm{lin}\) is \(\sigma^k+(k-1)\). Then \(\operatorname{sre}(w_\mathrm{lin}) = \sigma^k\). 
\end{lemma}
\begin{proof}[Generalized from the binary case proved in \cite{navarro2025smallest}]
    Every \(k\)-length substring occurs exactly once, therefore every substring \(x \in \Sigma^{k-1}\) is right-maximal and every right-extension \(xc, c\in \Sigma,\) is a super-maximal extension.\qed
\end{proof}

Having the value of \(\operatorname{sre}(w_\mathrm{lin})\), the next step is to find the size of a smallest suffixient set for \(w_\mathrm{lin}\$\). To this end, we generalize Lemma~\ref{lemma:suffixient_set_positions} for arbitrary \(\sigma\) values.

\begin{remark}\label{remark:generalized_suffixient_set_positions}
    Let \(w_\mathrm{lin}\$\) be a linearized and terminated \(\sigma\)-ary de Bruijn sequence of order \(k\). We can use Lemma~\ref{lemma:suffixient_set_positions} to extend the binary case. In the proof, replace \(2\) with an arbitrary \(\sigma\) value; similarly, list all possible right-extensions of a suffix \(s\) of length \(k-1\), one for each \(c \in \Sigma\). This procedure results in \(\sigma^k + 1\) positions in a smallest suffixient set.
\end{remark}

Next, we present a lower bound of the number of runs in the BWT of  \(w_\mathrm{lin}\$\) using a similar argument to Lemma~\ref{lemma:bwt_runs}.

\begin{lemma}\label{lemma:generalized_runs}
    Let \(w_\mathrm{lin}\$\) be a linearized and terminated \(\sigma\)-ary de Bruijn sequence of order \(k\). The lower bound of \(r\) in the BWT of \(w_\mathrm{lin}\$\) is \(\sigma^{k-1} (\sigma - 1) +1 \).
\end{lemma}
\begin{proof}
    To count the number of runs \(r\), we examine how the BWT of \(w_\mathrm{lin}\$\) behaves; we look at the last column \(L\) of the rotation matrix. 
    
    We start by looking at the rows independent of \(\$\). For each \(v \in \Sigma^{k-1}\), there are exactly \(\sigma\) suffixes that share the prefix \(v\), and such suffixes form a contiguous interval in the suffix array. In total, there are \(\sigma^{k-1}\) such intervals, one for each prefix \(v\). 

    Each such interval has \(\sigma\) suffixes, and their preceding characters are all pair-wise distinct elements of \(\Sigma\). Therefore, each interval contributes \(\sigma\) runs to the total.

    To minimize the number of runs, we merge each adjacent pair of intervals:

    \[
        \sigma^{k-1} \cdot \sigma - (\sigma^{k-1} - 1) = \sigma^{k-1}(\sigma - 1) + 1 \,.
    \]
    
    When we add back the \(\$\) having prefixes rows, the value of \(r\) increases: either we split the runs, or we maintain them unchanged, therefore:
    
    \[
        r \geq \sigma^{k-1} \cdot \sigma - (\sigma^{k-1} - 1) = \sigma^{k-1}(\sigma - 1) + 1 \,. \eqno \qed
    \]
\end{proof}

Finally, we prove that the ratio \(\chi / r\) for linearized and terminated de Bruijn sequences gets worse as \(\sigma\) increases.

\begin{theorem}
    For a linearized and terminated \(\sigma\)-ary de Bruijn sequence \(w_\mathrm{lin}\$\) of order \(k\), the ratio between the size \(\chi\) of a smallest suffixient set and the number of runs \(r\) in the BWT is bounded by
    \[
        \frac{\chi}{r} < \frac{\sigma}{(\sigma -1)} \,.
    \]
\end{theorem}
\begin{proof}
    From Remark~\ref{remark:generalized_suffixient_set_positions} and Lemma~\ref{lemma:generalized_runs}, we have the following:
    \[
        \frac{\chi}{r} \leq \frac{\sigma^k + 1}{\sigma^{k-1} (\sigma - 1) + 1} < \frac{\sigma}{\sigma-1} \,. \eqno \qed
    \]
\end{proof}

We conclude that de Bruijn sequences cannot approach the \(\chi \leq 2r\) bound for \(\sigma \geq 3\). 

\section{Final remarks}\label{sec:conclusion}
In our work, we expanded the relationship between the suffixient set and the Burrows--Wheeler Transform. In particular, we presented a general \(\sigma\)-ary family of strings that approaches the known \(\chi \leq 2r\) bound. 

This result, however, is still loose for small values of \(\sigma\). In light of this, we also showed a \(\sigma = 2\) result approaching the upper bound. To achieve it, we constructed a run-minimal family of de Bruijn strings using LFSRs, achieving the run-minimal pattern in the literature, if the sequence of values \(k\) for which the trinomial \(x^k +x + 1\) over \(\mathbb{F}_2\) is primitive (sequence OEIS A073639 \cite{OEIS:A073639}) is infinite.

We conjecture that the run-minimal family of strings \(\mathcal{M}\) is unique: when we apply the \(\$\) position finding algorithm from \cite{giuliani2021dollar} into a run-minimal pattern with any \(k\) that \(x^k + x + 1\) is not primitive over \(\mathbb{F}_2\), we fail to find any valid position for it.

We also showed that the construction for \(\sigma = 2\) cannot be generalized for \(\sigma \geq 3\), the de Bruijn structure gets more complex as \(\sigma\) increases. This infeasibility points towards different constructions for such \(\sigma\).

\bibliographystyle{splncs04}
\bibliography{sources}

\end{document}